\documentclass[preprint,12pt]{elsarticle}

\usepackage{color}
\usepackage{amsmath,amsfonts,latexsym,url}
\usepackage{bm}
\usepackage[colorlinks=true,citecolor=blue,urlcolor=magenta,linkcolor=true, hypertexnames=false]{hyperref}
\usepackage{xcolor}

\newtheorem{definition}{Definition}

\newtheorem{lemma}[definition]{Lemma}

\newtheorem{theorem}[definition]{Theorem}
\newtheorem{corollary}[definition]{Corollary}

\makeatletter
\def\iddots{\mathinner{\mkern1mu\raise\p@
\vbox{\kern7\p@\hbox{.}}\mkern2mu
\raise4\p@\hbox{.}\mkern2mu\raise7\p@\hbox{.}\mkern1mu}}
\makeatother

\newcommand{\PSres}{{\operatorname{PSres}}}
\newcommand{\Sres}{{\operatorname{Sres}}}

\newcommand{\chara}{{\operatorname{char}}}

\def\K{{\mathbb K}}
\def\cF{{\mathcal F}}
\def\cG{{\mathcal G}}
\def\M{{\mathsf M}}

\def\N{{\mathbb N}}
\def\Q{{\mathbb Q}}
\def\Z{{\mathbb Z}}

\def\twoFone#1#2#3#4{{_2F_1}\biggl(\begin{matrix}
  {#1}\kern.707em {#2}\\{#3}
\end{matrix}\,\bigg|\,#4\biggr)}

\newenvironment{proof}{
  \trivlist \item[\hskip \labelsep{\it Proof.}]}{\hfill\mbox{$\Box$}
  \endtrivlist}

\begin{document}

\begin{frontmatter}

\title{Subresultants of $(x-\alpha)^m$ and $(x-\beta)^n$, \\ Jacobi polynomials and complexity}

\author{A.~Bostan}
\address{Inria, Universit\'e Paris-Saclay, 1 rue Honor\'e d'Estienne d'Orves, 91120 Palaiseau, France}
\ead{alin.bostan@inria.fr}
\ead[url]{http://specfun.inria.fr/bostan}

\author{T.~Krick}
\address{Departamento de Matem\'atica, Facultad de
Ciencias Exactas y Naturales  and IMAS,
CONICET, Universidad de Buenos Aires,  Argentina} \ead{krick@dm.uba.ar}
\ead[url]{http://mate.dm.uba.ar/\~{}krick}

\author{A.~Szanto}
\address{Department of Mathematics, North Carolina State
University, Raleigh, NC 27695, USA}
\ead{aszanto@ncsu.edu}
\ead[url]{http://aszanto.math.ncsu.edu}

\author{M.~Valdettaro}
\address{Departamento de Matem\'atica, Facultad de
Ciencias Exactas y Naturales, Universidad de Buenos Aires, Argentina}
\ead{mvaldett@dm.uba.ar}
\ead[url]{http://cms.dm.uba.ar/Members/mvaldettaro/}

\begin{abstract}
In an earlier article together with Carlos D'Andrea~\cite{BDKSV17}, we
described explicit expressions for the coefficients of the order-$d$
polynomial subresultant of $(x-\alpha)^m$ and $(x-\beta)^n $ with respect to
Bernstein's set of polynomials $\{(x-\alpha)^j(x-\beta)^{d-j}, \, 0\le j\le
d\}$, for $0\le d<\min\{m, n\}$. The current paper further develops the study
of these structured polynomials and shows that the coefficients of the
subresultants of $(x-\alpha)^m$ and $(x-\beta)^n$ with respect to the monomial
basis can be computed in \emph{linear} arithmetic complexity, which is faster
than for arbitrary polynomials. The result is obtained as a consequence of the
amazing though seemingly unnoticed fact that these subresultants are scalar
multiples of Jacobi polynomials up to an affine change of variables.
\end{abstract}

\date{\today}
\end{frontmatter}

\section{Introduction}\label{sec:intro}

Let $\K$ be a field, and let $f=f_mx^m+\cdots + f_0$ and $g=g_nx^n+\cdots +
g_0$ be two polynomials in $\K[x]$ with $f_m\neq0$ and $g_n\neq 0$. Set $0\le
d< \min\{m,n\}$. The {\em order-$d$ subresultant} $\Sres_d(f,g)$ is the
polynomial in $\K[x]$ defined as
\begin{equation}\label{srs}
\Sres_d(f,g):= \det \begin{array}{|cccccc|c}
\multicolumn{6}{c}{\scriptstyle{m+n-2d}}&\\
\cline{1-6}
f_{m} & \cdots & &\cdots &f_{d+1-(n-d-1)} &x^{n-d-1}f& \\
&  \ddots & &&\vdots  & \vdots& \scriptstyle{n-d}\\
& & f_m& \dots &f_{d+1}&f& \\
\cline{1-6}
g_{n} &\cdots & &\cdots  &g_{d+1-(m-d-1)}  &x^{m-d-1}g&\\
&\ddots && &\vdots  &\vdots  &\scriptstyle{m-d}\\
&&g_{n} & \cdots &  g_{d+1} &g&\\
\cline{1-6} \multicolumn{2}{c}{}
\end{array},
\end{equation}
where, by convention, $f_\ell = g_\ell =0$ for $\ell<0$.

The polynomial $\Sres_d(f,g)$ has degree at most~$d$, and each of its
coefficients is {equal to} a minor of the Sylvester matrix of $f$ and $g$. {
In particular the coefficient of $x^d$, called the {\emph{principal
subresultant} of $f$ and $g$}, is given by
\begin{equation*}\label{PSres}
\PSres_d(f,g):= \det \begin{array}{|ccccc|c}
\multicolumn{5}{c}{\scriptstyle{m+n-2d}}&\\
\cline{1-5}
f_{m} & \cdots & &\cdots &f_{d-(n-d-1)}& \\
&  \ddots & & & \vdots& \scriptstyle{n-d}\\
& & f_m &\cdots & f_d& \\
\cline{1-5}
g_{n} &\cdots &  &\cdots &g_{d-(m-d-1)}&\\
&\ddots & &  &\vdots  &\scriptstyle{m-d}\\
& &g_{n} & \cdots &g_d&\\
\cline{1-5} \multicolumn{2}{c}{}
\end{array}.
\end{equation*}}

Subresultants were introduced implicitly by Jacobi~\cite{Jacobi1836} and
explicitly by Sylvester~\cite{Sylvester1839,Sylvester1840}; we refer
to~\cite{Loos83} and~\cite{GaLu03} for detailed historical
accounts\footnote{The Sylvester matrix was defined in~\cite{Sylvester1840},
and the order-$d$ subresultant was introduced
in~\cite{Sylvester1839,Sylvester1840} under the name of ``prime derivative of
the $d$-degree''. The term ``polynomial subresultant'' was seemingly coined by
Collins~\cite{Collins67}, and probably inspired to him by B\^ocher's
textbook~\cite[\S69]{Bocher1907} who had used the word ``subresultants'' to
refer to determinants of certain submatrices of the Sylvester matrix. Almost
simultaneously, Householder and Stewart~\cite{HS67,Householder68} employed the
term ``polynomial bigradients''. The principal subresultants were named ``Nebenresultanten'' (minor resultants) by Habicht~\cite{Habicht48}. The current terminology
\emph{principal subresultants} seems to appear for the first time in Collins'
paper~\cite{Collins74}.
}

Let $\M(n)$ denote the arithmetic complexity of degree-$n$ polynomial
multiplication in~$\K[x]$. Precisely, $\M(n)$ is an upper bound for the
total number of additions/subtractions and products/divisions in the base
field~$\K$ that are sufficient to compute the product of any two polynomials
in~$\K[x]$ of degree at most~$n$. It is classical, see
e.g.~\cite[Ch.~8]{GaGe13}, that $\M(n)= O(n \, \log n \log\log n)$ by using
FFT-based algorithms. For arbitrary polynomials $f,g\in \K[x]$ {of
degree~$n$}, the fastest known algorithms {are able to} compute in $O(\M(n)
\log n)$ arithmetic operations in $\K$ either one selected polynomial
subresultant $\Sres_d(f,g)$~{\cite{Reischert97,LR2001,Lecerf18}}, or all their
principal subresultants $\PSres_d(f,g)$ for $0\le
d<n$~\cite[Cor.~11.18]{GaGe13}. It is an open question whether this can be
improved to $O(\M(n))$, even for the {classical} resultant (the case $d=0$).

In this paper we present algorithms with \emph{linear} complexity for these
two tasks for the special family of polynomials considered in~\cite{BDKSV17},
namely $f=(x-\alpha)^m$ and $g=(x-\beta)^n$ in $\K[x]$, when
$\chara(\K)=0$ or $\chara(\K)\geq \max\{m,n\}$, and $\alpha\ne \beta \in
\K$ (note that when $\alpha=\beta$ there is nothing to compute since all
subresultants vanish). To our knowledge, we are exhibiting the first family of
``structured polynomials'' for which subresultants (and all principal
subresultants) can be computed in optimal arithmetic complexity.

Let us first observe that the resultant
$\Sres_0((x-\alpha)^m,(x-\beta)^n)=(\alpha-\beta)^{mn}$, which corresponds to
the case $d=0$, can be computed by binary powering in $O(\log(mn))$ arithmetic
operations in~$\K$. The general case is not so simple: for example the
particular case $d=1$ of \cite[Theorems~1.1 and~1.2]{BDKSV17} (see also
Theorem~\ref{thm:2} below) shows that, for $1<\min\{m,n\}$,
\begin{align*}		
	\Sres_1((x-\alpha)^m, (x-\beta)^n)&=
(\alpha - \beta)^{(m-1)(n-1)}  \Big(\binom{m+n-2}{m-1}   x \\ & \qquad - \binom{m+n-3}{m-1}\alpha -\binom{m+n-3}{n-1}\beta\Big).
\end{align*}
This identity implies that, from a computational perspective, there is already
a striking difference between the cases $d=0$ and $d=1$. Indeed, although the
term $(\alpha - \beta)^{(m-1)(n-1)}$ can be computed in $O(\log(mn))$
operations in $\K$, no algorithm with arithmetic complexity polynomial in
$\log(mn)$ is known for computing binomial coefficients such
as~$\binom{m+n-2}{m-1}$. However, the right-hand side of the previous identity
can be computed in $O(\min\{m,n\})$ operations (see Lemma~\ref{lem:binomial}
below), provided the characteristic of the base field~$\K$ is zero or large
enough. The main result of the current article extends this complexity
observation to arbitrary $1\le d<\min\{m,n\} $.

\begin{theorem} \label{thm:1}
Let $d,m, n\in \N$ with  $1\le d<\min\{m,n\}$
and let $\K$ be a field with $\chara(\K)=0$ or $\chara(\K)\geq \max\{m,n\}$, and $\alpha,\beta\in \K$ with $\alpha\ne \beta$.
Set
 \begin{equation*}\label{sk}
 \Sres_d((x-\alpha)^m,(x-\beta)^n)=\sum_{k=0}^d s_k\, x^k.
 \end{equation*}
Then,
\begin{enumerate} \item[\emph{(a)}] if $\chara(\K)=0$ or $\chara(\K)\ge m+n-d$, then $s_d\ne 0$ and all the coefficients
$s_k$ for $0\le k\le d$ can be computed using $O(\min\{m,n\}+\log(mn))$
arithmetic operations in~$\K$,
\item[\emph{(b)}] when $\chara(\K)= m+n-d-1$, the following equality holds in $\K$:
\[\Sres_d((x-\alpha)^m,(x-\beta)^n)= (-1)^{md}  (\alpha-\beta)^{(m-d)(n-d)+d} \]
and $\Sres_d((x-\alpha)^m,(x-\beta)^n)$ can be computed using  $O(\log(mn))$ arithmetic operations in~$\K$,
\item[\emph{(c)}]if $m+n-d-1> \chara(\K) \ge \max\{m,n\}$, then
$$\Sres_d((x-\alpha)^m,(x-\beta)^n)=0. $$
\end{enumerate}
\end{theorem}
We prove Theorem~\ref{thm:1} via an amazing (and seemingly previously
unobserved) close connection of the subresultants
$\Sres_d((x-\alpha)^m,(x-\beta)^n)$ with the classical family of orthogonal
polynomials known as {the} \emph{Jacobi polynomials}, introduced and studied
by Jacobi in his posthumous article~\cite{Jacobi1859}. This allows us to
produce a recurrence for the coefficients of the subresultant, which is
derived from the differential equation satisfied by the Jacobi polynomial, and
hence by the subresultant.

To express the polynomial subresultants $\Sres_d((x-\alpha)^m,(x-\beta)^n)$ as
Jacobi polynomials, let us recall \cite[Chapter 4]{Szeg75} that for any
$k,\ell,r\in \Z$ with $r \geq 0$, the Jacobi polynomial $P_r^{(k,\ell)}(x)$
can be defined in $\frac12 \mathbb{Z}[x]$, and thus also in $\K[x]$ for any
abstract field $\K$ with $\chara(\K) \neq 2$, in two equivalent ways:
\begin{itemize}
 \item by Rodrigues' formula
        \[ P_r^{(k,  {\ell})}(x):=\dfrac{(-1)^{r}}{2^r \, r!} (1-x)^{-k} (1+x)^{-{\ell}} \dfrac{\partial^r}{\partial x^r}\left[ (1-x)^{k+r} (1+x)^{ {\ell}+r} \right],\]
 \item as a hypergeometric sum:
         \[P_r^{(k, {\ell})}(x):=\sum_{j=0}^r \frac{(k+r-j+1)_j}{j!}\frac{({\ell}+j+1)_{r-j}}{(r-j)!} \left(\dfrac{x-1}{2}\right)^{r-j} \left(\dfrac{x+1}{2}\right)^j,\]
          where for any $a\in \Z$, $(a)_0:=1 $ and $(a)_j:=a(a+1)\cdots (a+j-1)$ for
$j\ge 1$ denotes the $j$th Pochhammer symbol, or the rising factorial,
of~$a$.
\end{itemize}
Our  next result asserts that the $d$-th subresultant
of $(x-\alpha)^m$ and $(x-\beta)^n$ coincides, up to an explicit
multiplicative constant and up to an affine change of variables, with the
Jacobi polynomial $P_d^{(-n,-m)}(x)$.
More precisely, for $\alpha\neq\beta$, we consider the following change  of variables in  the Jacobi polynomial
\begin{align}\label{integerPd}P_d^{(-n,-m)}&\left(\frac{(x-\alpha)+(x-\beta)}{\beta-\alpha}\right)=\\
&\sum_{j=0}^d\binom{n-d+j-1}{j}\binom{m-j-1}{d-j}\frac{(x-\alpha)^j(x-\beta)^{d-j}}{(\alpha-\beta)^d},
\nonumber
\end{align} and note that it belongs to
$\frac{1}{(\alpha-\beta)^d}\Z[x-\alpha,x-\beta]$ when we consider $\alpha$ and $\beta$ as distinct indeterminates over $\Z$.
 We denote by~$p_d$  its  coefficient of~$x^d$, for which we show in~\eqref{p=lead} below that
\begin{equation}\label{pd}
p_d= \frac{1}{(\alpha-\beta)^d} \binom{m+n-d-1}{d}.
\end{equation}
We also recall that, following the notation in Theorem~\ref{thm:1}, the principal subresultant $s_d:=\PSres_d((x-\alpha)^m, (x-\beta)^n )$  is the  coefficient   of $x^d$  in $\Sres_d((x-\alpha)^m, (x-\beta)^n)$, which by
 \cite[Proposition~3.3]{BDKSV17}  satisfies
\begin{equation}\label{sd}
s_d=(\alpha-\beta)^{(m-d)(n-d)}\prod_{i=1}^{d}r(i)\ \mbox{ with }\ r(i):=\frac{(i-1)!(m+n-d-i)!}{(m-i)!(n-i)!}.
\end{equation}
As a consequence, $s_d$ belongs to $\Q[\alpha-\beta] \cap \Z[\alpha,\beta]=
\Z[\alpha-\beta]$. In fact it is shown in~\cite[Theorem~1.1]{BDKSV17} that the
whole polynomial $\Sres_d((x-\alpha)^m,(x-\beta)^n)$ belongs to
$\Z[x-\alpha,x-\beta]$ (see also Lemma~\ref{lem:intsubres} below for an
independent proof). Denote by $Q_d^{(-n,-m)}$ the following polynomial
\[Q_d^{(-n,-m)}\left( \alpha, \beta, x\right) := \frac{s_d\cdot
P_d^{(-n,-m)}\left( \dfrac{(x-\alpha)+(x-\beta)}{\beta-\alpha}\right)}{p_d}.\]
Since $\alpha-\beta = (x-\beta) - (x-\alpha)$, the polynomial
$Q_d^{(-n,-m)}(\alpha,\beta,x)$ belongs a priori to $\Q[x-\alpha,x-\beta]$. We
will show that $Q_d^{(-n,-m)}(\alpha,\beta,x)$ actually coincides with
$\Sres_d((x-\alpha)^m,(x-\beta)^n) $ in $ \Z[x-\alpha,x-\beta]$ and we then
obtain, via the map $1_\Z\to 1_\K$, the following result:

\begin{theorem}\label{thm:2}
Let $\K$ be a field and $\alpha,\beta \in \K$ with $\alpha\ne \beta$. Set
$d,m, n\in \N$ with $0\le d<\min\{m,n\}$. Then, with the notation
in~\eqref{integerPd} and~\eqref{sd},
\begin{equation}\label{eq:Sres=Jac}	  \Sres_d((x-\alpha)^m, (x-\beta)^n)=
Q_d^{(-n,-m)}\left( \alpha, \beta, x\right).
\end{equation}
\end{theorem}

The key ingredient to prove Theorem~\ref{thm:1} will be to derive from
Theorem~\ref{thm:2} a second-order recurrence satisfied by the coefficients of
$\Sres_d((x-\alpha)^m,(x-\beta)^n)$ in the monomial basis, as follows:

\begin{theorem}\label{thm:3}
Let $\K$ be a field  and $\alpha,\beta\in \K$ with $\alpha\ne \beta$.
Set $d,m, n\in \N$ with  $0\le d<\min\{m,n\}$ and let
 \begin{equation*}\label{sk2}
 \Sres_d((x-\alpha)^m,(x-\beta)^n)=\sum_{k=0}^d s_k \,x^k.
 \end{equation*}
Then, when $\chara(\K)=0$ or $\chara(\K)\ge m+n-d$, for $s_{d+1}:=0$ and for
$s_d$ as defined in~\eqref{sd}, the following second-order linear recurrence
is satisfied by the coefficients $s_k$, for $ k= d-1,\dots,0$:
\begin{equation}\label{eq:rec}
s_k =\dfrac{ -(k+1)\Big(\big((n-k-1)\alpha+(m-k-1)\beta\big)s_{k+1}+
(k+2)\alpha\beta s_{k+2}\Big) }{(d-k)(m+n-d-k-1)}.
\end{equation}
\end{theorem}

Our next result concerns the complexity of the computation of all principal
subresultants $\PSres_d((x-\alpha)^m,(x-\beta)^n)$ for $0\le d<\min\{m,n\}$.
We note that the proof of this result is independent from our previous
results, as it is a consequence of a recurrence that is derived directly
from~\eqref{sd}. We give it here for sake of completeness of our complexity
results.

\begin{theorem}\label{thm:5}
Let $\K$ be a field, let $m, n\in \N$ and assume $\chara(\K)=0$ or
$\chara(\K)\ge m+n$. Let $\alpha,\beta \in \K$. Then one can compute all the
principal subresultants $\PSres_d((x-\alpha)^m,(x-\beta)^n)\in \K$ for $0\le d
< \min\{m,n\}$ using $O(\min\{m,n\}+\log(mn))$ operations in $\K$.
\end{theorem}

In the current article, we repeatedly use the crucial fact that, for
\emph{structured} algebraic objects, one can obtain improved complexity
results by using recurrence relations that these objects obey, rather than
just computing them independently. This is one of the strength of our results:
not only they provide nice formulae for the subresultants, but they also
exploit their particular structure in order to design efficient algorithms.

This work has an interesting story. While working on the paper \cite{BDKSV17},
we first realized that \cite[Theorems 1.1 and 1.2]{BDKSV17} (see
Theorem~\ref{th:old} below) implies the linear recurrence on the coefficients
of $\Sres_d((x-\alpha)^m,(x-\beta)^n)$ in the usual monomial basis described
in Theorem~\ref{thm:3}. This recurrence was initially found using a
computer-driven ``guess-and-prove'' approach, where the guessing part relied
on algorithmic {\emph{Hermite-Pad\'e approximation}}~\cite{gfun}, and where
the proving part relied on Zeilberger's \emph{creative telescoping}
algorithm~\cite{Zeilberger90,WiZe}. From this we derived a first proof of our
complexity result (Theorem~\ref{thm:1}). Shortly after that, by studying the
differential equation attached to this recurrence, we realized that it has a
basis of solutions of hypergeometric polynomials, which appeared to be Jacobi
polynomials. We have then obtained an indirect and quite involved proof of
Theorem~\ref{thm:2} and of Theorem~\ref{thm:3} based on manipulations of
hypergeometric functions, notably on the Chu-Vandermonde identity, much
inspired by an experimental mathematics approach. The proof that we choose to
present in this article is the shortest and the simplest that we could find.
It is chronologically the latest proof of our results, and the one which
provides the deepest structural insight. This proof was obtained by applying
some classical results and the fact that any polynomial that can be written as
a polynomial combination of $f$ and $g$ in $\K[x]$ with given degree bounds is
in fact a constant multiple of the subresultant of $f$ and $g$: we prove that
the Jacobi polynomial can indeed be expressed as such a combination of
$(x-\alpha)^m$ and $(x-\beta)^n$, and we determine the scalar multiple that
gives the subresultant. To conclude this introduction, we want to stress here
the importance of the interaction between computer science and classical
mathematics, which allowed us to guess and prove all our statements using the
computer, before finding a short and elegant human proof.

The paper is organized as follows: We first derive Theorems~\ref{thm:2}
and~\ref{thm:3} in Section~\ref{sec:thm3}. Section~\ref{sec:thm1} is dedicated
to the proof of Theorem~\ref{thm:1}, while in Section~\ref{sec:thm5} we prove
Theorem~\ref{thm:5}. Section~\ref{sec:previous} explains the connection of our
results with previous work, notably the relationship with classical results on
Pad\'e approximation. We conclude the paper with various remarks, experimental
results and perspectives in Section \ref{sec:final}.

A preliminary version of this work is part of the doctoral thesis of Marcelo
Valdettaro \cite{MVthesis}.

\paragraph*{Acknowledgements} \noindent We thank Christian Krattenthaler for
precious help with hypergeometric identities during an early stage\ of this
work, and to Mohab Safey El Din for generously sharing his subresultants
implementations with us. We are also grateful to the referees for helping us
substantially improve the presentation of our results. {T. Krick} and {M.
Valdettaro} were partially supported by \begin{normalsize} {ANPCyT}
{PICT-2013-0294}, {CONICET} {PIP-11220130100073CO}\end{normalsize} and
\begin{normalsize} {UBACyT} {2014-2017-20020130100143BA}. \end{normalsize} A.
Szanto was partially supported by the NSF grants CCF-1813340 and CCF-1217557.

\section{Proofs of Theorem~\ref{thm:2} and Theorem~\ref{thm:3}} \label{sec:thm3}

\subsection{{\it Proof of {Theorem~\ref{thm:2}}.}}

The proof of Theorem~\ref{thm:2} proceeds in 3 steps: (1) We prove the theorem
in the case when $\K$ has characteristic $0$. (2) We show, independently from
\cite{BDKSV17}, that $\Sres_d((x-\alpha)^m,(x-\beta)^n)$ belongs to $
\Z[x-\alpha,x-\beta]$ when we consider both polynomials $(x-\alpha)^m$ and
$(x-\beta)^n$ in $\Z[\alpha,\beta,x]$ for $\alpha,\beta$ new indeterminates
over $\Z$, which implies that $(\alpha-\beta)^d\,p_d$ divides
\[ s_d\cdot (\alpha-\beta)^d  P_d^{(-n,-m)}\left(
\dfrac{(x-\alpha)+(x-\beta)}{\beta-\alpha}\right)\ \mbox{ in } \
\Z[x-\alpha,x-\beta].\]
(Here we multiply both terms by $(\alpha-\beta)^d$ to guarantee that they are
both polynomials in $\Z[x-\alpha,x-\beta]$.) (3) We finally conclude that the
identity stated in Theorem~\ref{thm:2} holds in any characteristic via the map
$1_\Z\to 1_\K$.

We will need the next classical lemma, which follows e.g.
from~\cite[Lemmas~7.7.4 and~7.7.6]{Mishra} and was also a key ingredient in
\cite{BDKSV17}.

\begin{lemma}\label{lutil}  Let $ m,n\in \N$ and $f,\,g\in\K[x]$ {of
degrees} $m$ and $n$ respectively. Set $0\le d < \min\{m,n\}$ and assume
$\Sres_d(f,g)\neq 0$ has degree exactly $d$. If $\cF,\cG\in\K[x]$ with
$\deg(\cF)<n-d,\,\deg(\cG)<m-d$ are such that $h=\cF\,f+\cG\,g$ is a non-zero
polynomial in $\K[x]$ of degree at most $d$, then there exists
$\lambda\in\K\setminus \{0 \}$ satisfying \[h=\lambda\cdot\Sres_d(f,g).\]
\end{lemma}

\subsubsection{Proof of Theorem~\ref{thm:2} when $\chara(\K)=0$.}
\label{ssec:char0}
In this case $\Sres_d((x-\alpha)^m,(x-\beta)^n) $ has degree exactly $d$ by
Identity~\eqref{sd} since $\alpha\ne \beta$. We will then show that $h=
P_{d}^{(-n,-m)}\left(\dfrac{2x-\alpha-\beta}{\beta-\alpha}\right)$ satisfies
the conditions of Lemma~\ref{lutil} applied to $f=(x-\alpha)^m$ and $g=(x-\beta)^n$.\\
One can check {(or refer to \cite[Theorem 4.23.1]{Szeg75} to verify)} that the
polynomials $$P_{d}^{(-n,-m)}(z), \ (1+z)^m P_{n-d-1}^{(-n,m)}(z)\mbox{ and }
(1-z)^n P_{m-d-1}^{(n,-m)}(z),$$ all solve the {linear} differential equation
\[ (1-z^2)y''(z)+ \big((m+n-2)z-m+n\big)y'(z)+d(d+1-m-n)y(z)=0.\]
Substituting $z=\dfrac{2x-\alpha-\beta}{\beta-\alpha}$ in this differential
equation shows that the polynomials
\begin{align*}&y_1(x):=P_{d}^{(-n,-m)}\left(\dfrac{2x-\alpha-\beta}{\beta-\alpha}\right),\\ &y_2(x):=\left(\dfrac{2}{\beta-\alpha}\right)^m(x-\alpha)^m P_{n-d-1}^{(-n,m)}\left(\dfrac{2x-\alpha-\beta}{\beta-\alpha}\right)\ \mbox{ and } \\
&y_3(x):=\left(\dfrac{2}{\alpha-\beta}\right)^n(x-\beta)^n P_{m-d-1}^{(n,-m)}\left(\dfrac{2x-\alpha-\beta}{\beta-\alpha}\right),
\end{align*}
all solve the {linear} differential equation
\begin{align}\label{eq:diff}
(x-\alpha)(x-\beta)y''(x)+ &\big(\alpha(n-1)+\beta
(m-1)-(m+n-2)x\big)y'(x)\nonumber \\  & \qquad +d(m+n-d-1)y(x)=0.
\end{align}
Since the dimension of the solution space of this {second-order} linear
differential equation is 2, the three polynomials $y_1, y_2, y_3$ must be
linearly dependent {over~$\K$}. Now, it is well-known that {the} Jacobi
polynomials satisfy
 \begin{equation}\label{Jac1}
	P_r^{(k,{\ell})}(1)=\frac{(k+1)_r}{r!}\quad\mbox{and} \quad
	P_r^{(k,{\ell})}(-1)=(-1)^r\frac{({\ell}+1)_r}{r!}.
\end{equation}
This implies that
$y_2$ and $y_3$ are not linearly dependent over~$\K$ since
\begin{equation}\label{eq:y2}y_2(\beta)=2^mP_{n-d-1}^{(-n,m)}(1)=(-1)^{n-d-1}2^m\binom{n-1}{d} \ne 0 \ \mbox{ and } \  y_2(\alpha)=0,\end{equation}
while
\begin{equation}\label{eq:y3}y_3(\beta)=0  \ \mbox{ and } \  y_3(\alpha)=2^nP_{m-d-1}^{(n,-m)}(-1)= {2^n\binom{m-1}{d}} \ne 0.\end{equation}

Thus, there exist $A,B\in {\K}$ such that $y_1(x)=A \,y_2(x) + B \,y_3(x)$,
that is,
\begin{align}\label{bezout}
P_{d}^{(-n,-m)}\left(\frac{2x-\alpha-\beta}{\beta-\alpha}\right)=&A \left(\frac{2}{\beta-\alpha}\right)^m P_{n-d-1}^{(-n,m)}\left(\frac{2x-\alpha-\beta}{\beta-\alpha}\right) (x-\alpha)^m \\ & \nonumber \ +    B  \left(\frac{2}{\alpha-\beta}\right)^nP_{m-d-1}^{(n,-m)}\left(\frac{2x-\alpha-\beta}{\beta-\alpha}\right)(x-\beta)^n .
\end{align}
In addition
$P_{d}^{(-n,-m)}\left(\dfrac{2x-\alpha-\beta}{\beta-\alpha}\right)\ne 0$,
since
\begin{equation}\label{eq:spec} P_{d}^{(-n,-m)}( 1)= (-1)^d\binom{n-1}{d} \quad\mbox{and} \quad P_{d}^{(-n,-m)}( -1)= \binom{m-1}{d}.
\end{equation}
Moreover,
$\deg P_{d}^{(-n,-m)}\left(\frac{2x-\alpha-\beta}{\beta-\alpha}\right)\le d$, \
$\deg P_{n-d-1}^{(-n,m)}\left(\frac{2x-\alpha-\beta}{\beta-\alpha}\right) < n-d$    and  \\ $\deg P_{m-d-1}^{(n,-m)}\left(\frac{2x-\alpha-\beta}{\beta-\alpha}\right) < m-d$.
Therefore Lemma \ref{lutil} implies that there exists $\lambda \in \K$ such
that
\begin{equation}\label{Pd}
	P_{d}^{(-n,-m)}\left(\dfrac{2x-\alpha-\beta}{\beta-\alpha}\right)=\lambda \cdot \Sres_d((x-\alpha)^m, (x-\beta)^n).
\end{equation}
Thus, the left-hand side and right-hand side of this equality have the same
coefficient of $x^d$, which implies that $\lambda=p_d/s_d$. We now determine
$p_d$.

By Identity~\eqref{integerPd},
\begin{align}\label{p=lead}
p_d &= \frac{1}{(\alpha-\beta)^d}\sum_{j=0}^d
\binom{n-d+j-1}{j}\binom{m-j-1}{d-j}\nonumber  \\ &
 =\frac{1}{(\alpha-\beta)^d}\binom{m+n-d-1}{d},
\end{align}
where the second equation can be checked by thinking of a $d$-combination with
repetition from a set of size $m+n-2d$, written as a disjoint union of a
subset with $n-d$ elements and its complement with $m-d$ elements, computed by
adding, for $0\leq j \leq d$, the $j$-combination with repetition from the
first subset of size $n-d$ combined with the $(d-j)$-combination with
repetition from the second subset of size $m-d$.\\ Passing
$\lambda^{-1}=s_d/p_d$ to the left-hand side in Identity~\eqref{Pd} proves
Theorem~\ref{thm:2} when $\chara(\K)=0$.\hfill\mbox{$\Box$}

\subsubsection{Proof that $\Sres_d((x-\alpha)^m,(x-\beta)^n)$ belongs to $ \Z[x-\alpha,x-\beta]$.}\label{ssec:Z}

This result is already proved in \cite{BDKSV17}, but we give here an
independent proof because in Section~5.1 we will show the result in
\cite{BDKSV17} (see Theorem~\ref{th:old} below) and our Theorem~\ref{thm:2}
are equivalent.

\begin{lemma}\label{lem:intsubres} Set $d,m, n\in \N$ with  $0\le d<\min\{m,n\}$, and let $(x-\alpha)^m, (x-\beta)^n \in \Z[\alpha,\beta,x]$. Then
$$ \Sres_d((x-\alpha)^m, (x-\beta)^n)\in \Z[x-\alpha,x-\beta].$$
\end{lemma}
\begin{proof}
It is well-known from the matrix formulation of the subresultant that
$\Sres_d((x-\alpha)^m, (x-\beta)^n)\in \Z[\alpha,\beta,x]$.
Theorem~\ref{thm:2} gives us a way of writing $$ \Sres_d((x-\alpha)^m,
(x-\beta)^n)=(\alpha-\beta)^{(m-d)(n-d)} \sum_{j=0}^d
c_j(x-\alpha)^j(x-\beta)^{d-j} $$ where $c_j\in \Q$.

In particular, for $\alpha=0$ and $\beta=-1$, one has
on the one hand $$\Sres_d(x^m,(x+1)^n)=\sum_{j=0}^d c_j x^j (x+1)^{d-j},$$ with $c_j\in \Q$
while on the other hand $\Sres_d(x^m,(x+1)^n)=\sum_{k=0}^d a_k x^k$ with $a_k\in \Z$, $0\le k\le d$.
This means that
$$\sum_{j=0}^d c_j x^j (x+1)^{d-j}=\sum_{k=0}^d a_k x^k,$$ with $a_k\in \Z$ for $0\le k\le d$.
Comparing coefficients, we deduce that
$$a_k=  \sum_{j=0}^k \binom{d}{k-j} c_{j}, \quad 0\le k\le d,$$ i.e., that
$$\left(\begin{array}{c}a_0\\ \vdots \\a_d\end{array}\right)= \left(\begin{array}{cccc} 1 & & & \\
\binom{d}{1} & 1& & \\ \vdots &\vdots &\ddots & \\ \binom{d}{d} & \binom{d}{d-1} & \dots & 1\end{array}\right) \left(\begin{array}{c}c_0\\ \vdots \\c_d\end{array}\right).$$
We conclude that $c_j\in \Z$ for all $0\le j\le d$, since the $a_k$'s are
integer numbers and the transition matrix is an invertible integer matrix.
\end{proof}

\subsubsection{Concluding the proof of {Theorem~\ref{thm:2}}.}
We assume that  $\alpha$ and $\beta$ are distinct indeterminates over $\Q$. The theorem holds over the field $ \mathbb{Q}(\alpha,
\beta)$,
with both sides of equality~\eqref{eq:Sres=Jac} belonging to
$\Z[x-\alpha,x-\beta]$.
To prove  the
theorem for an arbitrary field~$\K$, and for distinct values~$\tilde{\alpha}$ and~$\tilde{\beta}$ in $\K$,
we apply a classical specialization
argument, using the ring homomorphism $\Z[x-\alpha, x- \beta] \to \K[x]$ which maps
$1_\Z\mapsto 1_\K, \alpha\mapsto \tilde{\alpha}, \beta\mapsto \tilde{\beta}$.

\subsection{Beyond Theorem~\ref{thm:2}}\label{ssec:beyond-thm:2}

An advantage of our proof of Theorem~\ref{thm:2} is that it also shows that
the unique polynomials $F_d$ and $G_d$ in $\K[x]$ of degrees respectively less
than $n-d$ and $m-d$ that are the coefficients of the {\em B\'ezout identity}
\begin{equation}\label{bezout2}\Sres_d((x-\alpha)^m,(x-\beta)^n)=F_d  \cdot (x-\alpha)^m+G_d
\cdot
(x-\beta)^n,
\end{equation}
are also (scalar multiples of) Jacobi polynomials, up to the same affine
change of variables. More precisely, we have:

\begin{corollary} \label{coro:cofact}
Let $\K$ be a field and $\alpha,\beta \in \K$ with $\alpha\ne \beta$. Set $d,m, n\in \N$ with  $0\le d<\min\{m,n\}$.
Then, the polynomials $F_d$ and $G_d$ defined in \eqref{bezout2} satisfy
\begin{align*} F_d &=\frac{(-1)^{n-1}s_d\, P_{n-d-1}^{(-n,m)}\left(\frac{(x-\alpha)+(x-\beta)}{\beta-\alpha}\right)}{(\beta-\alpha)^m \,p_d}, \\
 G_d &=\frac{(-1)^ns_d\, P_{m-d-1}^{(n,-m)}\left(\frac{(x-\alpha)+(x-\beta)}{\beta-\alpha}\right)}{(\beta-\alpha)^n\,p_d}.
\end{align*}
\end{corollary}

\begin{proof}
As in the proof of Theorem~\ref{thm:2} we first assume that $\K$ is a field of characteristic 0. By this theorem,  Identities \eqref{bezout2} and \eqref{bezout}, one has
\begin{align*}p_d\,F_d&= s_d\,A \left(\frac{2}{\beta-\alpha}\right)^m P_{n-d-1}^{(-n,m)}\left(\frac{2x-\alpha-\beta}{\beta-\alpha}\right), \\ p_d\,G_d&= s_d\, B  \left(\frac{2}{\alpha-\beta}\right)^nP_{m-d-1}^{(n,-m)}\left(\frac{2x-\alpha-\beta}{\beta-\alpha}\right).
\end{align*}
We now determine the values of $A$ and $B$.
By Identities~\eqref{eq:y2},~\eqref{eq:y3},~\eqref{bezout} and~\eqref{eq:spec}, we get
\begin{align*}&\binom{m-1}{d}= P_d^{(-n,-m)}(-1)=B \left(\frac{2}{\alpha-\beta} \right)^n P_{m-d-1}^{(n,-m)}(-1)(\alpha-\beta)^n\\ & \qquad \qquad \qquad  = 2^n {\binom{m-1}{d}B},\\
&(-1)^d\binom{n-1}{d}= P_d^{(-n,-m)}(1)=A \left(\frac{2}{\beta-\alpha} \right)^m P_{n-d-1}^{(-n,m)}(1)(\beta-\alpha)^m\\ & \qquad \qquad \qquad  =(-1)^{n-d-1}2^m {\binom{n-1}{d}A}.
\end{align*}
Therefore $A=\dfrac{(-1)^{n-1}}{2^m} $ and $B=\dfrac{1}{2^n}$. This proves the
statement when $\chara(\K)=0$. Finally, both sides in the equalities of
Corollary~\ref{coro:cofact} belong to
$\frac{1}{(\alpha-\beta)^{m+n-d-1}}\Z[\alpha,\beta,x]$ and so they specialize
well to a field of any characteristic via the map $1\mapsto 1_\K$.
\end{proof}

\subsection{{\it Proof of {Theorem~\ref{thm:3}}.}}
We now prove Theorem~\ref{thm:3}, which gives a recurrence
satisfied by the coefficients (in the monomial basis) of
$\Sres_d((x-\alpha)^m,(x-\beta)^n)$. The recurrence is inherited from
the differential equation~\eqref{eq:diff} satisfied by $P_d^{(-n,-m)}\left( \dfrac{(x-\alpha)+(x-\beta)}{\beta-\alpha}\right)$ in characteristic 0.

By Theorem~\ref{thm:2},
\begin{align}\label{eq:SP} \Sres_d((x-\alpha)^m, (x-\beta)^n)&=Q_d^{(-n,-m)}\left( \alpha, \beta, x\right) \nonumber \\
&= \frac{s_d}{p_d}\cdot P_d^{(-n,-m)}\left( \dfrac{(x-\alpha)+(x-\beta)}{\beta-\alpha}\right),\end{align}
 where
$ P_d^{(-n,-m)}\left( \dfrac{(x-\alpha)+(x-\beta)}{\beta-\alpha}\right)$ is the integer Jacobi polynomial described in Identity~\eqref{integerPd}, and \begin{align}\nonumber\frac{s_d}{p_d}&= (\alpha-\beta)^{(m-d)(n-d)+d}\frac{\prod_{i=1}^{d}r(i)}{\binom{m+n-d-1}{d}}\\
\label{s_d/p_d}& =  (\alpha-\beta)^{(m-d)(n-d)+d}\prod_{i=1}^d \dfrac{i!(m+n-d-i-1)!}{(m-i)!(n-i)!}.\end{align}
Therefore, the differential equation~\eqref{eq:diff} satisfied by the Jacobi polynomial is also satisfied by $s(x):=\Sres_d((x-\alpha)^m,(x-\beta)^n)$. We now show that this fact implies the statement.  We start with $$s(x)=\sum_{k=0}^d s_kx^k, \quad
s'(x)=\sum_{k=1}^{d}ks_k x^{k-1} \,
\mbox{ and } \, s''(x)= \sum_{k=2}^{d} k(k-1)s_{k}x^{k-2}.
$$
We then have  \begin{align*}
&(x-\alpha)(x-\beta)s''(x)= \sum_{k=2}^{d}k(k-1)s_kx^{k}-(\alpha+\beta)\sum_{k=2}^{d}k(k-1)s_kx^{k-1}\\ &\qquad \qquad \qquad \qquad +\alpha\beta\sum_{k=2}^{d}k(k-1)s_kx^{k-2}\\
&\qquad = \sum_{k=0}^{d}k(k-1)s_kx^{k}-(\alpha+\beta)\sum_{k=0}^{d-1}(k+1)ks_{k+1}x^{k}\\ &\qquad \qquad \qquad +\alpha\beta\sum_{k=0}^{d-2}(k+2)(k+1)s_{k+2}x^{k},
\end{align*}
\begin{align*}
&\left(\alpha(n-1)+\beta (m-1)-(m+n-2)x\right)s'(x)= -(m+n-2)\sum_{k=1}^{d}ks_kx^{k}\\ &\qquad \qquad \qquad +(\alpha(n-1)+\beta (m-1))\sum_{k=1}^{d}ks_kx^{k-1}\\
&= -(m+n-2)\sum_{k=0}^{d}ks_kx^{k}+(\alpha(n-1)+\beta (m-1))\sum_{k=0}^{d-1}(k+1)s_{k+1}x^{k},
\end{align*}
and
$$ d(m+n-d-1)s(x)=d(m+n-d-1)\sum_{k=0}^d s_kx^k.$$
Now we compare  the degree-$k$ coefficient in \eqref{eq:diff} for $k=0,\dots, d-1$:
\begin{align*}&\big(k(k-1)-(m+n-2)k+d(m+n-d-1\big)s_k +\big(- (\alpha+\beta)(k+1)k\\
& \quad + (\alpha(n-1)+\beta(m-1))(k+1)\big)s_{k+1}+\alpha\beta(k+2)(k+1)s_{k+2} =0.
\end{align*}
Therefore,
$$s_k =\dfrac{ -(k+1)\Big(\big((n-k-1)\alpha+(m-k-1)\beta\big)s_{k+1}+ (k+2)\alpha\beta s_{k+2}\Big) }{(d-k)(m+n-d-k-1)}.
$$
This proves the recurrence when $\chara(\K)=0$.
It is clear that the same recurrence also holds for fields $\K$ of
characteristic  $\ge m+n-d$ via the map $1_\Z \to 1_\K$ since in all the steps we are dividing only by natural numbers less than $m+n-d$.
{\hfill\mbox{$\Box$}}

\section{Proof of Theorem~\ref{thm:1}.} \label{sec:thm1}

\subsection{Proof of Theorem~\ref{thm:1}~\emph{(a)}.} \label{sec:thm1:case1}
We start with the following simple observation.

\medskip\begin{lemma} \label{lem:binomial}
Let $\K$ be a field, let $k,\ell \geq 0$ be integers and assume
$\chara(\K)=0$  or $ \chara(\K) > \min\{k,\ell\}$. Then  the (image
in~$\K$ of the) binomial coefficient $\binom{k+\ell}{k}$ can be computed in
$O(\min\{k,\ell\})$ arithmetic operations in~$\K$.
\end{lemma}

\begin{proof}
It is enough to use for $\binom{k+\ell}{k}$ the most economic of the
equivalent writings $(k+\ell)\cdots (k+1)/\ell!$ and $(\ell+k)\cdots
(\ell+1)/k!$.
\end{proof}	

The proof that one can compute all coefficients of the $d$-th subresultant of
$(x-\alpha)^m$ and $(x-\beta)^n$ in $O(\min\{m,n\}+\log(mn))$ operations
in~$\K$ when $\chara(\K)$ is either zero or larger than $m+n-d$ will be
derived from the recurrence~\eqref{eq:rec} described in Theorem~\ref{thm:3}.
The proof is algorithmic and proceeds in several steps.

We start with
$s_d=(\alpha-\beta)^{(m-d)(n-d)}\prod_{i=1}^{d}r(i)$, with $r(i)$ defined
in~\eqref{sd}, and observe that for the mentioned characteristics, $s_d\ne 0$ since $\alpha\ne \beta$.
\begin{itemize}
\item The term $(\alpha-\beta)^{(m-d)(n-d)}$ can be computed in
$O(\log(mn))$ arithmetic operations, by using binary powering. \item The
element $r(d)=(d-1)!\binom{m+n-2d}{m-d}$ can be computed in $O(\min\{m,n\})$
arithmetic operations by applying Lemma~\ref{lem:binomial}, and using that
$d<\min\{m,n\}$.
\item Thanks to the recurrence
$$r(i)=\frac{(m+n-d-i)}{i(m-i)(n-i)}r(i+1),$$ all $r(d-1),\dots,r(1)$ can be
deduced from $r(d)$ in $O(d)$ additional operations; then, computing
$r(1)\cdots r(d)$ also takes $O(d)$ operations.

Note that during the unrolling of the recurrence, the only divisions that
occur are by positive integers less than $\max\{m,n\}$,
legitimate in  $\K$ by the assumption on its characteristic.
\end{itemize}
This shows that $s_d$ can be computed using $O(\min\{m,n\}+\log(mn))$ arithmetic operations in~$\K$.
\begin{itemize}
\item
Starting from $s_{d+1}=0$ and $s_d$, we use the recurrence~\eqref{eq:rec}
 to compute $s_{d-1}, s_{d-2}, \ldots, s_0$ in $O(d)$ operations,
by adding $O(1)$ operations in~$\K$ for each of these $d$ terms.

Note that in this step only divisions by integers less than $m+n-d-1$ may occur, and all these elements are invertible in $\K$, by assumption.
\end{itemize}
In conclusion, all the coefficients $s_0, \ldots, s_d$ of
$\Sres_d((x-\alpha)^m,(x-\beta)^n)$ can be computed in
$O(\min\{m,n\}+\log(mn))$ operations in~$\K$, when $\chara(\K)=0$ or
$\chara(\K)\ge m+n-d$. \hfill\mbox{$\Box$}

\subsection{Proof of Theorem~\ref{thm:1}~\emph{(b)}. \label{sec:thm1:case2}}

We apply again the recurrence given by Theorem~\ref{thm:3}, in the
characteristic~0 case, to show that when $\chara(\K)=m+n-d-1$, the polynomial
subresultant $\Sres_d((x-\alpha)^m,(x-\beta)^n)) $ is actually a (non-zero)
constant in $\K$.

\begin{lemma}
Set $d,m, n\in \N$ with  $1\le d<\min\{m,n\}$ and let
 \begin{equation*}
 \Sres_d((x-\alpha)^m,(x-\beta)^n)=\sum_{k=0}^d s_k \,x^k \quad \in \ \Z[\alpha,\beta][x].
 \end{equation*}
Assume that   $m+n-d-1$ equals a prime number $p$.
Then $p\mid s_k$ in $\Z[\alpha,\beta]$ for $1\le k\le d$.
\end{lemma}
\begin{proof}
By applying Identity~\eqref{sd}, we first show that  $p\mid s_d$: clearly $p$ does not divide the denominator but  $p$ divides
 $(m+n-d-1)!$ which is in the numerator of $r(1)$. Therefore $p\mid \prod_{i=1}^d r(i)$ and $p\mid s_d$ (since $d\ge 1$).
 Observe that for $1\le k\le d-1$, the denominators that appear in the recurrence defining the sequence~$s_k$ in Theorem~\ref{thm:3} range from  $(d-1)(m+n-d-2)$ to $(m+n-2d)$, and thus none of them is divisible by $p=m+n-d-1$.
 Therefore, since $p\mid s_{d+1}$ and $p\mid s_d$, we inductively conclude that $p\mid  s_k$ for $1\le k\le d$.
\end{proof}

Via the map $1_\Z \to 1_\K$, we immediately deduce that  $s_d=\dots = s_1=0 $ in~$\K$, and therefore
$\Sres_d((x-\alpha)^m,(x-\beta)^n)\in \K$. We compute its value by specializing  Identity~\eqref{eq:SP} at $x=\alpha$, and thanks to \eqref{eq:spec} and \eqref{s_d/p_d}. Set  $p:=m+n-d-1=\chara (\K)$, then
 $\Sres_d((x-\alpha)^m,(x-\beta)^n)$ is equal to
\begin{align*}
&(\alpha-\beta)^{(m-d)(n-d)+d}\prod_{i=1}^d \dfrac{i!(m+n-d-i-1)!}{(m-i)!(n-i)!}P_d^{(-n,-m)}(-1)\\&=
(\alpha-\beta)^{(m-d)(n-d)+d}\prod_{i=1}^d \dfrac{i!(p-i)!}{(m-i)!(n-i)!}\binom{m-1}{d}\\& =
(\alpha-\beta)^{(m-d)(n-d)+d}\prod_{i=1}^d \dfrac{(i-1)!(p-i)!}{(m-i-1)!(n+i-d-1)!}
\\& =
(\alpha-\beta)^{(m-d)(n-d)+d}\prod_{i=1}^d \dfrac{\binom{p-1}{m-i-1}}{{\binom{p-1}{i-1}}}.
\end{align*}
It remains to show that the last product is equal to~$(-1)^{md}$ in~$\K$.
This is an immediate consequence of the following elementary lemma.

\begin{lemma}
$\displaystyle{\binom{p-1}{\ell} = (-1)^\ell}$ in $\K$,
for any~$0\leq \ell < p = \chara (\K)$.
\end{lemma}
\begin{proof}
By Fermat's little theorem we have  $(x-1)^{p-1} = (x-1)^{p}/(x-1) = (x^{p}-1)/(x-1) = x^{p-1}+\cdots+1$ in $\K[x]$. Thus, the coefficient $(-1)^{\ell} \binom{p-1}{\ell}$ of $x^{\ell}$ in $(x-1)^{p-1}$
is equal to 1 in $\K$.
\end{proof}
Finally, by the previous lemma, the following equalities hold in $\K$:
$$\prod_{i=1}^d \frac{\binom{p-1}{m-i-1}}{{\binom{p-1}{i-1}}}=\prod_{i=1}^d \frac{(-1)^{m-i-1}}{(-1)^{i-1}} = (-1)^{md}.$$
This concludes the proof of Theorem~\ref{thm:1}~{(b)}.
\hfill\mbox{$\Box$}

\subsection{Proof of Theorem~\ref{thm:1}~\emph{(c)}.} \label{sec:thm1:case3}

This  non-obvious fact follows for instance from Theorem~\ref{thm:2}. We know by Identity~\eqref{eq:SP} in the characteristic~0 case that
$$\Sres_d((x-\alpha)^m, (x-\beta)^n)=
 \frac{s_d}{p_d}\cdot P_d^{(-n,-m)}\left( \dfrac{(x-\alpha)+(x-\beta)}{\beta-\alpha}\right)$$
 where
$ P_d^{(-n,-m)}\left( \dfrac{(x-\alpha)+(x-\beta)}{\beta-\alpha}\right)$ is the integer polynomial described in Identity~\eqref{integerPd}, and $$\frac{s_d}{p_d}= (\alpha-\beta)^{(m-d)(n-d)+d}\prod_{i=1}^d \dfrac{i!(m+n-d-i-1)!}{(m-i)!(n-i)!}.$$
We note that the denominator in this  last term does not vanish  in the mentioned characteristics while the numerator equals $0$, since it is a multiple of $(m+n-d-2)!$ for $d\ge 1$. We conclude the proof of Theorem~\ref{thm:1}(c) via the map $1_\Z \to 1_\K$.  \hfill\mbox{$\Box$}

\smallskip \noindent {\bf Remark.} Notice that Theorem~\ref{thm:1}(c)
also follows from  Theorem~\ref{thm:1}(b) and from Collins' fundamental theorem of
subresultants (\cite[\S4]{Collins73}, see also~\cite[\S2]{Habicht48} and~\cite[Theorem 1]{Collins67}) which states that for
an arbitrary field $\K$ and arbitrary $f, g\in \K[x]$, the subresultants and the
Euclidean remainder sequence of $f$ and~$g$ are closely related: if
$A_1:=f, A_2:=g, A_3,\ldots, A_\ell$ is an Euclidean polynomial remainder sequence
of $f$ and $g$ with $\deg(A_k)=n_k$ for $1\leq k \leq \ell$, then
there exist $c_1, \ldots, c_\ell, d_1, \ldots, d_\ell \in \K^\times$  such that
\begin{eqnarray*}
&&\Sres_{n_k}(f,g)= c_k \cdot A_k, \;\; \Sres_{n_{k-1}-1}(f,g)= d_k \cdot A_{k}, \text{ and } \\
&&\Sres_d(f,g)=0 \text{ for } n_{k}<d<n_{k-1}-1,
\end{eqnarray*}
for all $1\leq k \leq \ell$.
In particular, if two nonzero subresultants $\Sres_e(f,g)$ and $\Sres_{e'}(f,g)$  have the same degree for some $e'<e$, then they are constant multiples of each other, and all the intermediate subresultants $\Sres_{d}(f,g)$ are zero for $e'<d<e$. In
our situation, with $\max\{m,n\}\leq p:=\chara(\K)< m+n-d - 1$, and
$f=(x-\alpha)^m, g=(x-\beta)^n$ in $\K[x]$ with $\alpha\neq \beta$, we
have that $\Sres_{0}(f,g) \in\K^\times$ and also,
by Theorem~\ref{thm:1}(b), that $\Sres_{m+n-p-1}(f,g) \in\K^\times$.
Therefore, $\Sres_d (f,g)=0$ for $1\leq d < m+n-1-p$,
which reproves part (c) of Theorem~\ref{thm:1}.

\section{Proof of Theorem~\ref{thm:5}} \label{sec:thm5}
With the notation $r(i):=\dfrac{(i-1)!(m+n-d-i)!}{(m-i)!(n-i)!}$ introduced in \eqref{sd}, we have:
\begin{equation*}\label{eq:cd}
	\PSres_d((x-\alpha)^m,(x-\beta)^n) =(\alpha-\beta)^{(m-d)(n-d)}\prod_{i=1}^{d}r(i).
\end{equation*}
While in previous sections $d$ was considered as a fixed value, in this
section we view it as variable. Therefore, in order to avoid confusion, we
write $r_d(i):=r(i)$, to
emphasize also its dependence on~$d$. For all integers $d\geq 1$, we define
\[c(d):= \prod_{i=1}^{d}r_d(i)\] and note that it is an integer number, as
mentioned in the introduction, although the terms~$r_d(i)$ are not all
integers. We also set $c(0):=1$. 

The key observation for what follows is contained in the next lemma.

\begin{lemma}\label{lem:uv}
Let $\K$ be a field with
 $\chara(\K)=0$ or $\chara(\K)\ge m+n$.
  Set $u(d) := {c(d)}/{c(d-1)}$ for $1\le d< \min\{m,n\}$ and
$v(d) := {u(d+1)}/{u(d)}$ for $1\le d\le \min\{m,n\}-2$. Then,
for $1\le d\le \min\{m,n\}-2$,
\begin{equation} \label{eq:uk}
	v(d) =
\frac{d(m-d)(n-d)(m+n-d)}{(m+n-2d-1)(m+n-2d)^2(m+n-2d+1)}.
\end{equation}
\end{lemma}

\begin{proof}
We have that $u(1)=c(1)= \binom{m+n-2}{m-1}$  and for $d\ge 2$,
$$\frac{r_d(i)}{r_{d-1}(i)} = \frac{1}{(m+n-d-i+1)}.$$ Therefore
\begin{align*}u(d)& = \frac{c(d)}{c(d-1)}
=\frac{\prod_{i=1}^{d}r_d(i)}{\prod_{i=1}^{d-1}r_{d-1}(i)}
=r_d(d) \cdot \prod_{i=1}^{d-1} \frac{r_d(i)}{r_{d-1}(i)}
\\
&= (d-1)! \binom{m+n-2d}{m-d} \cdot \prod_{i=1}^{d-1} \frac{1}{m+n-d-i+1}.
\end{align*}
Hence \begin{align*}v(d) & =\frac{u(d+1)}{u(d)}\\ & =
d\frac{(m-d)(n-d)}{(m+n-2d-1)(m+n-2d)}
\cdot
\frac{(m+n-d)}{(m+n-2d)(m+n-2d+1)},
\end{align*}
which is the desired expression.

Note that the only numbers that appear in the denominators of $u(d)$ and
of $v(d)$ are products of integers of absolute value less than $m+n$, which
are invertible in $\K$ by the assumption on the characteristic of $\K$.
\end{proof}

Based on Lemma~\ref{lem:uv}, we now design an algorithm that computes all
principal subresultants $\PSres_d((x-\alpha)^m,(x-\beta)^n)$ with $1\le
d<\min\{m,n\}$ in {$O(\min\{m,n\}+\log(mn))$} operations in $\K$, thus proving
Theorem~\ref{thm:5}.

\begin{itemize} \item First, $v(1), \ldots, v(\min\{m,n\}-2)$ are computed by~using~\eqref{eq:uk} in $O(1)$ arithmetic operations each, for a total of
$O(\min\{m,n\})$ operations in~$\K$.
\item Then, $u(1), \ldots, u(\min\{m,n\}-1)$ are determined, by computing
$u(1) := \binom{m+n-2}{m-1}$ using Lemma~\ref{lem:binomial}, in
$O(\min\{m,n\})$ arithmetic operations in~$\K$, and by computing iteratively
$u(d) = u(d-1)\cdot v(d-1)$, for $ 2\le d<\min\{m,n\}$, in $O(\min\{m,n\})$
operations in $\K$.
\item Next we compute the elements $c(1), \ldots, c(\min\{m,n\}-1)$ iteratively by
$c(d) = u(d) \cdot c(d-1)$ for $1\le d<\min\{m,n\}$, in $O(\min\{m,n\})$
operations in $\K$.
\end{itemize}
At this stage, it remains to compute all the powers
$h(d):=(\alpha-\beta)^{(m-d)(n-d)}$ for $0\le d<\min\{m,n\}$, and finally to
output $\PSres_d((x-\alpha)^m,(x-\beta)^n)= c(d) \cdot h(d)$, for $0\le
d<\min\{m,n\}$. This is done as follows.

\begin{itemize}
\item
First, all the elements $\gamma(d):=(\alpha-\beta)^{2d+1-m-n}$, for
$d<\min\{m,n\}$, are computed using $O(\log(m+n) + \min\{m,n\})$ operations in
$\K$. This can be done by computing $\gamma(0) := (\alpha-\beta)^{1-m-n}$ by
binary powering, then unrolling the recurrence $\gamma(d+1) := (\alpha-\beta)^2
\cdot \gamma(d)$ for $d<\min\{m,n\}-1$.
\item
Next, $h(0) :=(\alpha-\beta)^{mn}$ is computed by binary powering, and then
all $h(d)$, for $1\le d<\min\{m,n\}$, by repeated products using $h(d+1) :=
\gamma(d) \cdot h(d)$, for a total cost of $O(\log(mn) + \min\{m,n\})$
operations in $\K$.
\item
Finally, we compute and return the values
$\PSres_d((x-\alpha)^m,(x-\beta)^n)= c(d) \cdot h(d)$, for $0\le
d<\min\{m,n\}$, using $O(\min\{m,n\})$ operations in~$\K$.
\end{itemize}
Adding up the various arithmetic costs proves Theorem~\ref{thm:5}.
\hfill\mbox{$\Box$}

\section{{Connections to previous results}} \label{sec:previous}
Theorem~\ref{thm:2} is closely connected to some previous results. First we discuss the connection to the work \cite{BDKSV17}.
Second, we explain the relationship of the present work to classical results on  {\em  Pad\'e approximation}.

\subsection{Connection with~\cite{BDKSV17}}

We show that the expression for the subresultant obtained in~\cite{BDKSV17},
though not expressed in terms of Jacobi polynomials, is  equivalent to the one
in Theorem~\ref{thm:2}. First, let us recall the main results
of~\cite{BDKSV17}.

\begin{theorem} \label{th:old} \cite[Theorems 1.1 and 1.2]{BDKSV17}\\ Let $\K$ be a field and $\alpha,\beta\in \K$.  Set $d,m, n\in \N$ with  $0\le d<\min\{m,n\}$. Then,
{\small \[
\Sres_d((x-\alpha)^m,(x-\beta)^n) =
{(\alpha-\beta)^{(m-d)(n-d)} \sum_{j=0}^d c_j(m,n,d)(x-\alpha)^j(x-\beta)^{d-j}},
\]}
where the coefficients $c_0(m,n,d), \ldots, c_d(m,n,d)$
are defined by
$$
c_0(m,n,d)=	
\displaystyle{\prod_{i=1}^{d}}\dfrac{(i-1)!\,(m+n-d-i-1)!}{(m-i-1)!(n-i)!},$$ and
$$c_j(m,n,d)=
\frac{\binom{d}{j}\binom{n-d+j-1}{j}}{\binom{m-1}{j}} \, c_0(m,n,d), \quad \text{for} \quad 1\le j\le d.$$
\emph{(}Here  {$c_0(m,n,0)=1$}, following the convention that an empty product equals~1.\emph{)}\\
Moreover, for $0\le j\le d$,  $c_j(m,n,d)\in \Z \text{ or } \Z/p\Z$  if  $\chara(\K)=0$  or $\chara(\K)=p$, respectively.
\end{theorem}

\noindent
\paragraph*{Proof that Theorems~\ref{th:old}  and~\ref{thm:2} are equivalent}\label{10}

We want to prove that
\begin{equation}\label{th2=th9}
(\alpha-\beta)^{(m-d)(n-d)}\sum_{j=0}^d c_j(m,n,d)(x-\alpha)^j(x-\beta)^{d-j}=\frac{s_d\,P_d^{(-n,-m)}\left( \dfrac{2x-\alpha-\beta}{\beta-\alpha}\right)}{p_d},\end{equation}
where
 $$c_j(m,n,d)= \frac{\binom{d}{j}\binom{n-d+j-1}{j}}{\binom{m-1}{j}}
 \displaystyle{\prod_{i=1}^{d}}\dfrac{(i-1)!\,(c-i)!}{(m-i-1)!(n-i)!}$$
 for $c:=m+n-d-1$.

By
 \eqref{s_d/p_d} the right-hand side of \eqref{th2=th9} equals
$$(\alpha-\beta)^{(m-d)(n-d)+d}\prod_{i=1}^d \dfrac{i!(c- {i})!}{(m-i)!(n-i)!}\,
 P_d^{(-n,-m)}\left( \dfrac{2x-\alpha-\beta}{\beta-\alpha}\right),$$
where by \eqref{integerPd},

\begin{align*}
{{(\alpha-\beta)^d}}&  P_d^{(-n,-m)}\left(\frac{2x-\alpha-\beta}{\beta-\alpha}\right)\\
&= \sum_{j=0}^d \binom{n-d+j-1}{j}\binom{m-j-1}{d-j}(x-\alpha)^j(x-\beta)^{d-j}.\end{align*}
Thus,   we only need to verify that
\begin{align*}\binom{n-d+j-1}{j}&\binom{{m}-j-1}{d-j} \prod_{i=1}^d\frac{i!(c-i)!}{(m-i)!(n-i)!}\\&=\frac{\binom{d}{j}\binom{n-d+j-1}{j}}{\binom{m-1}{j}}
\prod_{i=1}^d\frac{(i-1)!(c-i)!}{(m-i-1)!(n-i)!},\end{align*}
i.e. after simplification, that
$$\frac{(m-1)!}{(m-d{-1})!}\prod_{i=1}^d\frac{i!}{(m-i)!}=d!\prod_{i=1}^d\frac{(i-1)!}{(m-i-1)!},$$
which trivially holds.
{\hfill\mbox{$\Box$}}

\subsection{Connection with Pad\'e approximation}

In this subsection we show that Theorem~\ref{thm:2} and
Corollary~\ref{coro:cofact} are also equivalent to classical descriptions of
some Pad\'e approximants via Gauss hypergeometric functions.

The starting point is a theorem due to Pad\'e~\cite{Pade01}, stating that the
$[m/n]$ Pad\'e approximation in $\mathbb{C}(x)$ to $(1-x)^k$ is the ratio of
hypergeometric functions
\begin{equation} \label{eq:pade}
	\frac{\,_2F_1(-m,-k-n;-m-n;x)}{\,_2F_1(-n,k-m;-m-n;x)}.
\end{equation}
That result had been previously obtained, by different methods and under several additional assumptions, by Laguerre~\cite{Lag85}
and Jacobi~\cite{Jacobi1859}.
See also~\cite[Eq.~(Pad\'e 5), p. 252]{Perron13},~\cite[p.~65]{Baker75}, \cite{Iserles79} and Theorem 4.1
in~\cite{GGZ12}.

There is also a well-known connection between subresultants and Pad\'e approximants (c.f. \cite[Corollary 5.21]{GaGe13}): the $[m/n]$ Pad\'e approximation in
$\mathbb{C}(x)$ to $(1-x)^k$, for integer $k\ge m$, equals
\begin{equation} \label{eq:pade1}
	\frac{\Sres_m(x^{m+n+1}, (1-x)^k)}{G_m(x^{m+n+1}, (1-x)^k)}\,=\,(-1)^k \frac{\Sres_m(x^{m+n+1}, (x-1)^k)}{G_m(x^{m+n+1}, (x-1)^k)},
\end{equation}	
where $G_m:=G_m(x^{m+n+1}, (x-1)^k)$ is the polynomial coefficient of degree $\le n$ in the B\'ezout expression
$$\Sres_m(x^{m+n+1}, (x-1)^k)= F_m \cdot x^{m+n+1} + G_m \cdot (x-1)^k.$$
Identity~\eqref{eq:pade} implies that $$ \frac{\,_2F_1(-m,-n-k;-m-n;x)}{\,_2F_1(-n,k-m;-m-n;x)}=(-1)^k\frac{\Sres_m(x^{m+n+1}, (x-1)^k)}{G_m(x^{m+n+1}, (x-1)^k)}.$$
We showed earlier that the fact that
  $x^{m+n+1}$ and $(x-1)^k$  are coprime  polynomials  implies that $\deg(\Sres_m(x^{m+n+1}, (x-1)^k))=m$,  and  it is also immediate  to verify that
$\Sres_m(x^{m+n+1}, (x-1)^k)$ and $G_m(x^{m+n+1}, (x-1)^k)$ are coprime. Therefore, since the degree of $$\,_2F_1(-m,-k-n;-m-n;x)=\sum_{i=0}^m(-1)^i \binom{m}{i}\frac{(-k-n)_i}{(-m-n)_i}x^i, $$ equals $m$, one derives that there exists a non-zero $\lambda\in \mathbb{C}$ such that
\begin{align*}
&\Sres_m(x^{m+n+1}, (x-1)^k)= \lambda \cdot \,_2F_1(-m,-k-n;-m-n;x),\\
&G_m(x^{m+n+1}, (x-1)^k)= (-1)^k \lambda \cdot \,_2F_1(-n,k-m;-m-n;x).
\end{align*}
Here, $\lambda$ can be computed by comparing  the leading coefficients of \\$\Sres_m(x^{m+n+1}, (x-1)^k)$ and
$\,_2F_1(-m,-k-n;-m-n;x)$:
\begin{align*}
\lambda&= (-1)^m\frac{(k+n-m)!(m+n)!}{ (k+n)!n!} \,\PSres_m(x^{m+n+1}, (x-1)^k)\\
&= (-1)^{(n+1)(k-m)+m}\prod_{i=1}^m \frac{(i-1)!(k+n-i)!}{(k-i)!(m+n-i)!},
\end{align*}
by Identity~\eqref{sd}.

Now,
according to \cite[(1.6)]{EMOT1953}, see also \cite[(1.5)]{Koornwinder84}:
\begin{align*}
&\,_2F_1(-m,-k-n;-m-n;x)= \frac{1}{\binom{m+n}{m}}P_m^{(-k,-m-n-1)}\left( {2x-1}\right), \\
&\,_2F_1(-n,k-m;-m-n;x)= \frac{1}{\binom{m+n}{m}}P_n^{(k,-m-n-1)}\left({2x-1}\right),
\end{align*}
while, according to our Theorem~\ref{thm:2} and Corollary~\ref{coro:cofact},
\begin{align*}
&\Sres_m(x^{m+n+1}, (x-1)^k)=\mu\, P_m^{(-k,-m-n-1)}(2x-1),\\
&G_m(x^{m+n+1}, (x-1)^k)= (-1)^k \overline\mu \,P_n^{(k,-m-n-1)}(2x-1),
\end{align*}
for \begin{align*}
\mu & :=  (\alpha-\beta)^{(m-d)(n-d)+d}\prod_{i=1}^d \dfrac{i!(m+n-d-i-1)!}{(m-i)!(n-i)!} \quad \mbox{and}\\
\overline \mu&:=(-1)^{(n+1)(k-m)+m}\prod_{i=1}^m  \frac{i!(k+n-i)!}{(k-i)!(m+n+1-i)!}.\end{align*}
This shows the equivalence of the results  for $\alpha=0,\beta=1$, since
$ \lambda=\binom{m+n}{m}\overline \mu$.\\
In order to deduce Theorem~\ref{thm:2} and Corollary~\ref{coro:cofact} for any $\alpha,\beta$ we apply the usual changes of variables formulas that can be found in the now classical book \cite{AJ2006}:
\begin{align*}
& \Sres_d(f (x-\alpha),g(x-\alpha))=\Sres_d(f,g)(x-\alpha),\\
&\Sres_d(f(\gamma x),g(\gamma x))=\gamma^{mn-d(d+1)}\Sres_d(f,g)(\gamma x).
\end{align*}
Therefore,
\begin{align*}
& \Sres_d((x-\alpha)^m,(x-\beta)^n)=\Sres_d(x^m,(x-(\beta-\alpha))^n)(x-\alpha),\\
&\Sres_d(x^m,(x-\gamma)^n)(\gamma x)=\frac{1}{\gamma^{mn-d(d+1)} } \,\Sres_d ((\gamma x)^m, (\gamma x-\gamma )^n)\\ &
 \qquad \qquad = \frac{1}{\gamma^{mn-d(d+1)} } \,\Sres_d (\gamma^mx^m, \gamma^n(x-1)^n)\\ & \qquad \qquad
= \frac{\gamma^{m(n-d)+n(m-d)}}{\gamma^{mn-d(d+1)}}\,\Sres_d(x^m,(x-1)^n)\\
& \qquad \qquad
 =   \gamma^{(m-d)(n-d)+d}\,\Sres_d(x^m,(x-1)^n).
\end{align*}
Hence, since we have just proven that $\Sres_d(x^m,(x-1)^n)= \tilde \mu\, P_d^{-n,-m}(2x-1)$
for $\tilde \mu = \prod_{i=1}^d \frac{i!(m+n-d-i-1)!}{(m-i)!(n-i)!}$, we deduce that
$$\Sres_d(x^m,(x-(\beta-\alpha))^n)((\beta-\alpha) x)= \tilde\mu\,(\beta-\alpha)^{(m-d)(n-d)+d} P_d^{-n,-m}(2x-1),$$
which implies that
$$\Sres_d(x^m,(x-(\beta-\alpha))^n)(x)= \tilde\mu\,(\beta-\alpha)^{(m-d)(n-d)+d} P_d^{-n,-m}\left(2\left(\frac{x}{\beta-\alpha}\right)-1\right).$$
We conclude with
\begin{align*}
& \Sres_d((x-\alpha)^m,(x-\beta)^n)=\Sres_d(x^m,(x-(\beta-\alpha))^n)(x-\alpha)\\
&\qquad\qquad  = \tilde\mu\,(\beta-\alpha)^{(m-d)(n-d)+d} P_d^{-n,-m}\left(2\left(\frac{x-\alpha}{\beta-\alpha}\right)-1\right)\\
& \qquad\qquad =   \tilde\mu\,(\beta-\alpha)^{(m-d)(n-d)+d} P_d^{-n,-m}\left(\frac{2x-\alpha-\beta}{\beta-\alpha}\right),
\end{align*}
as stated in Theorem~\ref{thm:2}.

Note that similar arguments allow to deduce $G_d((x-\alpha)^m,(x-\beta)^n)$ from $G_d(x^m,(x-1)^n)$.

\section{Final remarks} \label{sec:final}

\subsection{Fast computation of cofactors} One can use similar ideas as in the
proof of Theorem~\ref{thm:1} in order to compute the cofactors $F_d(x)$ and
$G_d(x)$ in Corollary~\ref{coro:cofact} using $O( \max\{m,n\}+\log(mn) )$
arithmetic operations in $\K$,  when
$\chara(\K)=0$ or $\chara(\K)\geq \max\{m,n\}$.
More precisely, we have the following result, whose proof is omitted:

\begin{theorem} \label{thm:13}
Let $d,m, n\in \N$ with  $1\le d<\min\{m,n\}$
and let $\K$ be a field with $\chara(\K)=0$ or $\chara(\K)\geq \max\{m,n\}$, and $\alpha,\beta\in \K$ with $\alpha\ne \beta$.
 Let $F_d$ and $G_d$ be as defined in \eqref{bezout2}.  Then,
\begin{enumerate} \item[\emph{(a)}]  if $\chara(\K)=0$ or $\chara(\K)\ge m+n-d$, then all  the coefficients of $F_d$ and $G_d$
 can be computed using $O(\max\{m,n\}+\log(mn))$
arithmetic operations in~$\K$,
\item[\emph{(b)}] when $\chara(\K)= m+n-d-1$, the following equalities hold in $\K$
\begin{eqnarray*}
 F_d &=& (-1)^{dm+1} (\alpha-\beta)^{(m-d-1)(n-d-1)} (x-\alpha)^{n-d-1} , \\
        G_d &=& (-1)^{dm} (\alpha-\beta)^{(m-d-1)(n-d-1)} (x-\beta)^{m-d-1} ,
\end{eqnarray*}
and  the coefficients of $F_d$ and $G_d$ can be computed using  $O(\max\{m,n\}+\log(mn))$ arithmetic operations in~$\K$,
\item[\emph{(c)}] if $m+n-d-1> \chara(\K) \ge \max\{m,n\}$ then
$$F_d=G_d=0. $$
\end{enumerate}
\end{theorem}

\subsection{Comparison with generic algorithms} As mentioned in the
introduction, the fastest algorithms for subresultants of polynomials of
degree at most~$n$ have arithmetic complexity $O(\M(n) \log n)$, where $\M(n)$
denotes the arithmetic complexity of degree-$n$ polynomial
multiplication~\cite{Reischert97,LR2001,Lecerf18}. These algorithms can
compute either one selected polynomial subresultant, or all principal
subresultants. Using FFT-based algorithms for polynomial
multiplication~\cite[Ch.~8]{GaGe13}, their complexity $O(\M(n) \log n)$
becomes $O(n \log^2 n \log\log n)$, which is quasi-linear up to
polylogarithmic factors. These algorithms are generic in the sense that they
apply to arbitrary polynomials, and they work in any characteristic.

The algorithms described in the current article are specific to very
structured polynomials, namely pure powers of linear polynomials, and they
achieve purely linear arithmetic complexity in their maximum degree~$n$. They
also compute either one selected polynomial subresultant, or all principal
subresultants, but they are restricted to characteristic zero or large enough.
The reason is that they require divisions, which is the price to pay for
optimality. We leave as an open question whether purely linear arithmetic
complexity can be also achieved in arbitrary characteristic.

Another interesting difference is that, while classical algorithms for the
order-$d$ subresultant spend more time when~$d$ is small (typically, the
resultant computation, corresponding to $d=0$, is the most expensive), our
algorithms spend less time when~$d$ is small. For more on practical
comparisons, see \S\ref{ssec:practical}.

\subsection{Algorithmic optimality} The complexity result $O(
\min\{m,n\}+\log(mn))$ is quasi-optimal for Theorem~\ref{thm:5}, since the
size of the output is $\min\{m,n\}$. On the other hand, the complexity result
$O( \min\{m,n\}+\log(mn) )$ for Theorem~\ref{thm:1} is not optimal when $d$ is
small compared to $m$ and~$n$. A natural question is whether an algorithm of
arithmetic complexity $O( d+\log(mn))$ may exist. While this is true for
$d=0$, we believe that this is unlikely for $d\geq 1$, and moreover we
suspect that there is no algorithm for Theorem~\ref{thm:1} with arithmetic
complexity polynomial in both $d$ and $\log(mn)$. Otherwise, we could in
particular compute the first principal subresultant
\[\PSres_1((x-\alpha)^m,(x-\beta)^n)=(\alpha-\beta)^{(m-1)(n-1)}
\binom{m+n-2}{m-1},\]
in arithmetic complexity \emph{polynomial in $\log(mn)$}. This does not seem
plausible, since it would imply in particular that the central binomial
coefficient $\binom{2N}{N}$ could be computed using an arithmetic complexity
polynomial in $\log N$. Although no proof exists, this is generally believed
to be impossible.

\subsection{Fast factorials} It is possible to further improve some of our
complexity results by using Strassen's algorithm~\cite{Strassen76} for the
computation of~$N!$ in arithmetic complexity $O(\M(\sqrt{N}) \log N)$, which
becomes quasi-linear in~$\sqrt{N}$ when FFT-based algorithms are used for
polynomial multiplication. For instance, for fixed $d$, the principal
subresultant $\PSres_d((x - \alpha)^m, (x - \beta)^n)$ can be computed using
fast factorials in $$O(d + \log(mn) + \M(\sqrt{\min\{ m-d,n-d}\}) \log \min\{
m-d,n-d \} ),$$ operations in $\K$. The same cost can also be achieved for the
computation of the whole polynomial subresultant $\Sres_d((x - \alpha)^m, (x -
\beta)^n)$ in Theorem~\ref{thm:1}.

\subsection{Bit complexity} We have only discussed arithmetic complexity. When
$\K$ is a finite field, this is perfectly realistic, since arithmetic
complexity reflects quite well the running time of the algorithms. When $\K$
is infinite, for instance when $\K=\Q$, assuming operations in $\K$ at unit
cost is not realistic anymore, so studying bit complexity becomes a much more
pertinent model. Over $\K=\Q$, our algorithms in Sections~\ref{sec:thm1}
and~\ref{sec:thm5} have very good complexity behaviors in this model too.
Indeed, they only involve binary powering, computation of factorials and
binomials, unrolling of recurrences, which can be computed in quasi-optimal
bit complexity. This is confirmed by the timings in Tables~\ref{tab:timings}
and~~\ref{tab:timings-bis}, which appear to be indeed quasi-linear in the
output size.

\begin{table}[t]
  \centering
\begin{tabular}{|c|c|c|r|r|r|}
  \hline
  $\#$ & $(\alpha,\beta)$ & $(m,n,d)$ & {\sf Generic 1} & {\sf New 1} & {\sf Output size}\\
  \hline
  {\sf T1} & (10, 11) & (121, 92, 32) & 0.164 & 0.001 & $112\,125$ \\
  {\sf T2} & (13, 17) & (196, 169, 84)   & 5.439 & 0.002 & $2\,463\,994$ \\
  {\sf T3} & (12, 19) & (227, 245, 87)   & 23.543 & 0.006 & $6\,996\,907$ \\
  {\sf T4} & (12, 14) & (483, 295, 203)  & 71.613 & 0.011 & $11\,869\,930$ \\
  {\sf T5} & (10, 7) & (715, 694, 290)   & 2112.891 & 0.092 & $123\,580\,220$ \\
  {\sf T6} & (8, 4) & (1917, 1532, 805)   & --- & 1.227 & $1\,982\,541\,397$ \\
  {\sf T7} & (8, 4) & (2409, 3833, 1261)  & --- & 7.847 & $10\,745\,238\,510$ \\
  {\sf T8} & (3, 2) & (7840, 6133, 3510)  & --- & 40.983 & $45\,784\,567\,320$ \\
  \hline
\end{tabular}
\caption{
 Comparative timings (in seconds) for the computation of the polynomial subresultants $\Sres_d((x-\alpha)^m,(x-\beta)^n)$, on several instances of $(\alpha,\beta)\in\mathbb{Q}^2$ and $(m,n,d)\in\mathbb{N}^3$, using a generic subresultant algorithm implemented in the {\sf RegularChains} package
(column~{\sf Generic 1}), versus the specialized algorithm described in Section~\ref{sec:thm1} (column~{\sf New 1}).
    All examples were run on the same machine, with the latest version of {\sf Maple}. For entries marked with a ---, the computations were aborted after more than 17 hours, with  all available memory (150 Gb of RAM) consumed.
The last column displays the bit size of the output.
}
  \label{tab:timings}
\end{table}

\begin{table}[t]
  \centering
\begin{tabular}{|c|r|r|r|r|r|}
  \hline
  $\#$ & {\sf Generic 2}  & {\sf Output size G2} & {\sf New 2} & {\sf Output size N2}\\
  \hline
  {\sf T1} & 0.011 & $3\,297$ & 0.001 & $201\,764$ \\
  {\sf T2}   & 0.071 & $28\,739$ & 0.005 & $5\,113\,012$ \\
  {\sf T3}   & 0.281 & $79\,253$ & 0.030 & $14\,744\,328$ \\
  {\sf T4}  & 0.306 & $57\,633$ & 0.034 & $24\,875\,833$ \\
  {\sf T5}   & 8.921 & $423\,993$ & 0.905 & $249\,854\,978$ \\
  {\sf T6}   & 211.895 & $2\,458\,114$ & 12.578 & $4\,187\,207\,983$ \\
  {\sf T7}  & 1992.231 & $8\,511\,770$ & 83.145 & $21\,885\,019\,390$ \\
  {\sf T8}  & 15627.306 & $13\,035\,552$ &  237.423 & $57\,964\,587\,220$ \\
  \hline
\end{tabular}
\caption{
Comparative timings (in seconds) for the computation of the principal subresultants $\PSres_d((x-\alpha)^m,(x-\beta)^n)$, on the instances {\sf T1}--{\sf T8} from Table~\ref{tab:timings}, using a generic subresultant algorithm implemented in {\sf C}
(column~{\sf Generic 2}), versus the specialized algorithm described in Section~\ref{sec:thm5} implemented in {\sf Maple} (column~{\sf New~2}).
Column {\sf Output size {\sf G2}} displays the bit size of the integer $\PSres_d((x-\alpha)^m,(x-\beta)^n)$ computed by~{\sf Generic 2}.
Timings displayed in column~{\sf New~2} correspond to the computation of
all $\PSres_k((x-\alpha)^m,(x-\beta)^n)$ for $0\leq k < \min \{ m,n \}-1$.
Column {\sf Output size {\sf N2}} displays the bit size of the $\min \{ m,n \}$ integers computed by~{\sf New 2}.
}
  \label{tab:timings-bis}
\end{table}

\subsection{Practical issues} \label{ssec:practical} The algorithms described
in this article have not only a good theoretical complexity, but also a good
practical efficiency. We performed some experimental comparisons in {\sf
Maple}, between an implementation of our specialized algorithm in
Section~\ref{sec:thm1} and a generic subresultant algorithm available in the
package {\sf
RegularChains}\footnote{\url{http://www.regularchains.org/index.html}}. As
expected, our algorithm is much faster, since it exploits the special
structure of the input polynomials.

Table~\ref{tab:timings} displays some timings for computing
$\Sres_d((x-\alpha)^m,(x-\beta)^n)$, for various random choices of
$\alpha,\beta,m,n$ and~$d$. Even for moderate degrees $m, n$, the specialized
algorithm is about thousands of times faster. For higher degrees, the generic
algorithm becomes quite slow, while the specialized algorithm has a very
satisfactory speed.

We also implemented in {\sf Maple} the algorithm in Section~\ref{sec:thm5},
and this time we compared it, on the same examples as in
Table~\ref{tab:timings},  with an algorithm written in~{\sf C} by Mohab
Safey El Din. The experimental results are displayed in
Table~\ref{tab:timings-bis}. Once again, the specialized algorithm is faster
than the generic algorithm.

\subsection{Subresultants for other structured polynomials} The question
addressed in this article is a particular case of a much broader topic, the
design of efficient algorithms for \emph{structured polynomials}.

Preliminary results indicate that, for many polynomials whose coefficients
satisfy linear recurrences, their subresultants have coefficients that also
obey such recurrences; this leaves hope that their computation can be
performed in linear time. We plan to study such generalizations in a future
work.

For the time being, we performed promising experiments for subresultants of generalized Laguerre polynomials~\cite[\S5.1]{Szeg75}, defined by
\[ L_n^{(\alpha)} (x) = \sum_{i=0}^n
{n+\alpha \choose n-i} \frac{(-x)^i}{i!}, \]
and on classical Hermite polynomials~\cite[\S5.5]{Szeg75}, defined by
\[H_{2n} (x) = (2n)! \sum_{m=0}^{n} \frac{(-1)^m}{m!(2n - 2m)!} (2x)^{2n - 2m}.\]

\bigskip
\bibliographystyle{alpha}
\def\cprime{$'$} \def\cprime{$'$} \def\cprime{$'$}

\end{document}